\documentclass{llncs}

\usepackage{amsmath}
\usepackage{multirow}
\usepackage{bbding}
\usepackage{tabularx}
\usepackage{array}      %Useful for a better array control

\usepackage{wrapfig}

\usepackage{algorithm}
\usepackage{algpseudocode}
\algnotext{EndFor}
\algnotext{EndIf}

\usepackage{tikz} 
\usetikzlibrary{positioning}

\usepackage{subfigure}
\usepackage{makeidx}

\newcommand{\eqspace}{\quad \quad} 
\newtheorem{observation}{Observation}

\def\lengtharrow{40mm}

\newcommand{\boitecentre}[1]{\begin{minipage}[t][][t]{\lengtharrow}{\hbox to 
\lengtharrow{\hfill#1\hfill}}\end{minipage}}
\newcommand{\flechegauche}[1]{$\stackrel{\text{#1}}{\hbox to 
\lengtharrow{\leftarrowfill}}$}
\newcommand{\flechedroite}[1]{$\stackrel{\text{#1}}{\hbox to 
\lengtharrow{\rightarrowfill}}$}
\newcolumntype{C}{>{\scriptsize}{c}}

\newcommand\blfootnote[1]{%
  \begingroup
  \renewcommand\thefootnote{}\footnote{#1}%
  \addtocounter{footnote}{-1}%
  \endgroup
}

\begin{document}

%\mainmatter 

\title{Distance-bounding facing both mafia and distance frauds: Technical 
report $\star$}

\author{Rolando Trujillo-Rasua\inst{1}\and Benjamin Martin\inst{2}\and Gildas 
Avoine\inst{2,3}}
\institute{ Interdisciplinary Centre for Security, Reliability and Trust \\
University of Luxembourg \\  \and 
Universit\'e catholique de Louvain\\Information Security Group, Belgium \\ \and 
INSA 
Rennes, IRISA UMR 6074, 
France\\}

\maketitle
%%%%%%%%%%%%%%%%%%%%%%%%%%%%%%%%%%%%%%%%%%%%%%%%%%%%%%%%%%%%%%%%
\thispagestyle{plain}
\begin{abstract}
Contactless technologies such as RFID, NFC, and sensor networks are vulnerable to mafia and distance frauds. Both frauds aim at passing an authentication protocol by cheating on the actual distance between the prover and the verifier. To cope these security issues, distance-bounding protocols have been designed. However, none of the current proposals simultaneously resists to these two frauds without requiring additional memory and computation. The situation is even worse considering that just a few distance-bounding protocols are able to deal with the inherent background noise on the communication channels. This article introduces a noise-resilient distance-bounding protocol that resists to both mafia and distance frauds. The security of the protocol is analyzed with respect to these two frauds in both scenarios, namely noisy and noiseless channels. Analytical expressions for the adversary's success probabilities are provided, and are illustrated by experimental results. The analysis, performed in an already existing framework for fairness reasons, demonstrates the undeniable advantage of the introduced lightweight design over the previous proposals.

\keywords{Authentication, distance-bounding, relay attack, mafia fraud, distance fraud, noise.}
\end{abstract}

\section{Introduction}
\blfootnote{This document contains content accepted for publication at the IEEE 
Transactions on Wireless Communications~\cite{TMA2014}.}

%%%%%%%%%%%%%%%%%%%%%%%%%%%%%% introduction to relay attacks%%%%%%%%%%%%%%%%%%%%%%%%%%

A \emph{mafia fraud} is a man-in-the-middle attack applied against an authentication protocol where the adversary simply relays the exchanges without neither manipulating nor understanding them~\cite{AvoineBKLM-2011-jcs}. The earliest version of this attack was introduced by Conway in 1976 and is known as the \emph{chess grandmaster problem}~\cite{citeulike:1223195}. In this problem, a little girl is able to compete with two chess grandmasters during a postal chess game, where she transparently relays the moves between the two grandmasters. She eventually wins a game or draws both. 

%%%%%%%%%%%%%%%%%%%%%%% The relay attack in the crypographic field -- mafia frau %%%%%%%%%%%%%%%

In modern cryptography, mafia frauds can typically be used against authentication protocols. The adversary relays the messages between the prover and the verifier, who think they communicate together, while there is an adversary in the middle. This so-called mafia fraud was actually suggested by Desmedt, Bengio and Goutier in 1987~\cite{Desmedt:1987:SUS:646752.704723} to defeat the Fiat-Shamir protocol~\cite{FiatS-1986-crypto}.

One of the most promising field to apply the mafia fraud is the contactless technology, especially Radio Frequency IDentification (RFID) and Near Field Communication (NFC) where the devices answer to any solicitation without explicit agreement of their holder. Some attacks have already been performed against both RFID and NFC systems~\cite{Francis:2010:PNP:1926325.1926331,Hancke:2008:ATD:1352533.1352566}. Nevertheless, mafia fraud is not limited to contactless technologies, it also threats other technologies such as smartcards~\cite{DrimerM-2007-usenix} and e-voting~\cite{OrenW-2009-eprint}.

%%%%%%%%%%%%%%%%%%%%%%% Introducing other type of frauds %%%%%%%%%%%%%%%%%%%%%%%%%%%%%

Two other attacks related to the mafia fraud exist: the \emph{terrorist fraud} and the \emph{distance fraud}. The distance fraud only involves a malicious prover, who cheats on his distance to the verifier. It was first introduced by Brands and Chaum~\cite{Brands:1994:DP:188307.188361}, and comes from the distance measuring process used to defeat the mafia fraud. The terrorist fraud is a variant of the mafia fraud where the prover is malicious and actively helps the adversary to succeed the attack~\cite{BengioBDGQ-1991-crypto}. No solution exists yet to avoid this exotic fraud, which is not addressed in this paper. Additional countermeasure must actually be considered to thwart this fraud.

%%%%%%%%%%%%%%%%%%%%%%% distance-bounding as a countermeasure %%%%%%%%%%%%%%%%%%%%%%%%%%%%%%%%%%%%%%%%%

As mentioned above, a distance measuring process can mitigate the mafia and distance frauds. To that aim, Brands and Chaum~\cite{Brands:1994:DP:188307.188361} proposed the \emph{distance-bounding protocols} (DB protocols). The distance estimation relies on the measurement of the Round-Trip-Time (RTT) of single bit exchanges. Considering the physical impossibility to travel faster than the speed of light, RTT bounds the distance between the parties. Several distance-bounding protocols have been proposed~\cite{AvoineBKLM-2011-jcs}. However, none of the current DB protocols are lightweight and resistance to both mafia and distance frauds. Furthermore, just a few of them are able
to deal with the inherent background noise of the communication
channel. 

%%%%%%%%%%%%%%%%%%%%%%%5 Contributions %%%%%%%%%%%%%%%%%%%%%%%%%%%%%%%%%%%

\noindent\textbf{Contribution.} In this paper we introduce a novel DB protocol that significantly reduces the success probability of an adversary capable of mounting both mafia and distance frauds. Our protocol does not rely on computationally expensive primitives, has a very low memory requirement, and is noise-resilient. Therefore, it is efficient and suitable for extremely low resources devices. We provide analytical and experimental results that together show the superiority of our proposal w.r.t. to previous ones. %To the best of our knowledge, there does not exist ......bla bla bla

\noindent\textbf{Organization.} Further below Section~\ref{section:review} presents a brief background about DB protocols. Section~\ref{sec:rationality} explains the rationality behind our proposal and Section~\ref{section:protocol} introduces and details the proposal. Sections~\ref{sec:mafia} and~\ref{sec:distance} are dedicated to the resistance of the protocol to mafia and distance frauds respectively. Section~\ref{section:noise} describes our noise resiliency mechanism. Section~\ref{section:result} provides comparative results with several DB protocols in both scenarios the free-noise case and the noisy case. Finally, Section~\ref{sec:conclusion} draws the conclusions.

\section{Background on distance-bounding}\label{section:review}

The first lightweight DB protocol was proposed by Hancke and Kuhn's~\cite{Hancke:2005:RDB:1128018.1128472} in 2005. Its simplicity and suitability for resource-constrained devices have promoted the design of other DB protocols based on it~\cite{AvoineT-2009-isc,KimA-2011-ieeetwc,Trujillo-Rasua:2010:PDP:1926325.1926352}. All these protocols share the same design: (a) there is a slow phase\footnote{In DB protocols, a \emph{fast phase}, which generally consists on $n$ rounds, is a phase where the verifier computes RTTs. Otherwise, we say that it is a \emph{slow phase}.} where both prover and verifier generate and exchange nonces, (b) the nonces and a keyed cryptographic hash function are used to compute the answers to be sent (resp. checked) by the prover (resp. verifier). Below, we provide the main characteristics of each of these protocols, especially the technique they use to compute the answers. 

\noindent\textbf{Hancke and Kuhn's protocol~\cite{Hancke:2005:RDB:1128018.1128472}.} The answers are extracted from two $n$-bit registers such that any of the $n$ $1$-bit challenges determines which register should be used to answer. 

\noindent\textbf{Avoine and Tchamkerten's protocol~\cite{AvoineT-2009-isc}.} 
Binary trees are used to compute the prover answers: the verifier challenges 
define the unique path in the tree, and the prover answers are the vertex value 
on this path. There are several parameters impacting the memory consumption and 
the resistance to distance and mafia frauds: $l$ the number of trees and $d$ 
the depth of these trees. It holds $d\cdot l=n$, where $n$ is the number of 
rounds in the fast phase. The larger $d$, the better the frauds resistance and 
the larger the memory consumption. 

\noindent\textbf{Trujillo-Rasua, Martin and Avoine's protocol~\cite{Trujillo-Rasua:2010:PDP:1926325.1926352}.} This protocol is similar to the previous one, except that it uses particular graphs instead of trees to compute the prover answers.

\noindent\textbf{Kim and Avoine's protocol~\cite{KimA-2011-ieeetwc}.} This 
protocol, closer to the Hancke and Kuhn's 
protocol~\cite{Hancke:2005:RDB:1128018.1128472} than~\cite{AvoineT-2009-isc} 
and~\cite{Trujillo-Rasua:2010:PDP:1926325.1926352}, uses two registers to 
define the prover answers. An important additional feature is that the prover 
is able to detect a mafia fraud thanks to \emph{predefined challenges}, that 
is, challenges known by both prover and verifier. The number of predefined 
challenges impacts the frauds resistance: the larger, the better the mafia 
fraud resistance, but the lower the resistance to distance fraud.
%a predefined one, and a random one. 

There exist other DB protocols with different designs and computational 
complexities (\emph{e.g.,} protocols based on signatures and/or a final extra 
slow 
phase~\cite{Brands:1994:DP:188307.188361,Singelee:2007:DBN:1784404.1784415}). 
However, they are beyond the scope of this article that focuses on lightweight 
protocols only. The interested reader could refer to~\cite{AvoineBKLM-2011-jcs} 
for more details.

\section{Rationality of our proposal}\label{sec:rationality}

Being resistant to both mafia and distance frauds is the primary goal of a DB protocol. An important lower-bound for both frauds is $\left(\frac12\right)^n$, which is the probability of a naive adversary who answers randomly to the $n$ verifier challenges during the fast phase. However, this resistance is hard to attain for lightweight DB protocols. Therefore, our aim is to design a protocol that is close to this bound for both mafia and distance frauds, without requiring costly operations or an extra final slow phase. We also aim to reach the additional property of being noise-resilient. Below, the intuitions that lead to our design are explained for each of the three considered properties.

\noindent\textbf{Mafia fraud.} Among the DB protocols without final slow phase, those achieving the best mafia fraud resistance are round-dependent~\cite{AvoineT-2009-isc,KimA-2011-ieeetwc,Trujillo-Rasua:2010:PDP:1926325.1926352}. The idea is that the correct answer at the $i$th round should depend on the $i$th challenge and also on the ($i-1$) previous challenges. Our proposal also uses a round-dependent technique, the proposed construction is significantly simpler than those proposed in~\cite{AvoineT-2009-isc,KimA-2011-ieeetwc,Trujillo-Rasua:2010:PDP:1926325.1926352}, though.

\noindent\textbf{Distance fraud.} As in mafia fraud, the best protocols in term of distance fraud are round-dependent. However, round-dependency by means of predefined challenges as in the Kim and Avoine's construction~\cite{KimA-2011-ieeetwc} fails to properly resist to distance fraud. Intuitively, as more control over the challenges the prover has, the lower the resistance to distance fraud is. For this reason, our proposal allows the verifier to have full and exclusive control over the challenges.

\noindent\textbf{Noise-resiliency.} Round-dependent protocols can hardly work in noisy environments. A noise in a given round might affect all the subsequent rounds and thus, these rounds becomes useless from the security point of view. Therefore, in order to deal with noise, round-dependent protocols should be able to detect the noisy-rounds so that the prover responses can be checked considering these noise occurrences. To the best of our knowledge, our protocol is the first round-dependent DB protocol able to detect such a noisy-rounds with a high level of accuracy thanks to the simplicity of its design.

\section{Proposal}\label{section:protocol}

This section describes the DB protocol introduced in the paper. Initialization, 
execution, and decision steps are presented below and a general view is 
provided in Figure~\ref{fig:2b}.

\subsubsection{Initialization.} The prover ($P$) and the verifier ($V$) agree 
on 
(a) a security parameter $n$, (b) a timing bound $\Delta t_{\textrm{max}}$, (c) 
a pseudo random function  $PRF$ that outputs $3n$ bits, (d) a secret key $x$.

\subsubsection{Execution.} The protocol consists of a slow phase and a fast 
phase.

\subsubsection{Slow Phase.} $P$ (respectively  $V$) randomly picks a nonce 
$N_{P}$ (respectively $N_{V}$) and sends it to $V$ (respectively  $P$). 
Afterwards, $P$ and $V$ compute $PRF\left(x,N_{P},N_{V}\right)$ and divide the 
result into three $n$-bit registers $Q$, $R^0$, and $R^1$. Both $P$ and $V$ 
create the function $f_Q : \mathcal{S} \rightarrow \{0,1\}$ where $\mathcal{S}$ 
is the set of all the bit-sequences of size at most $n$ including the empty 
sequence. The function $f_Q$ is parameterized with the bit-sequence $Q = 
q_1\dots q_n$, and it outputs $0$ when the input is the empty sequence. For 
every non-empty bit-sequence $C_i = c_1\dots c_i$ where $1 \leq i \leq n$, the 
function is defined as $f_Q(C_i) = \bigoplus_{j = 1}^{i}(c_j \wedge q_j)$.

\subsubsection{Fast Phase.} In each of the $n$ rounds, $V$ picks a random 
challenge $c_{i} \in_{R} \{0,1\}$, starts a timer, and sends $c_{i}$ to $P$. 
Upon reception of $c_i$, $P$ replies with $r_{i} =R_{i}^{c_i} \oplus f_Q(C_i)$ 
where $C_i = c_1...c_i$. Once $V$ receives $r_i$, he stops the timer and 
computes the round-trip-time $\Delta t_i$.

\subsubsection{Decision.} If $\Delta t_{i}<\Delta 
t_{\textrm{max}}$ and $r_{i}=R_{i}^{c_{i}}\oplus f_Q(C_i)$  $\forall \ i \in 
\{1, 2, ..., n\}$ then the protocol succeeds\footnote{$\Delta t_{\textrm{max}}$ 
is a system parameter that implicitly 
represents the maximum allowed distance between the prover and the verifier.}.
\begin{figure*}
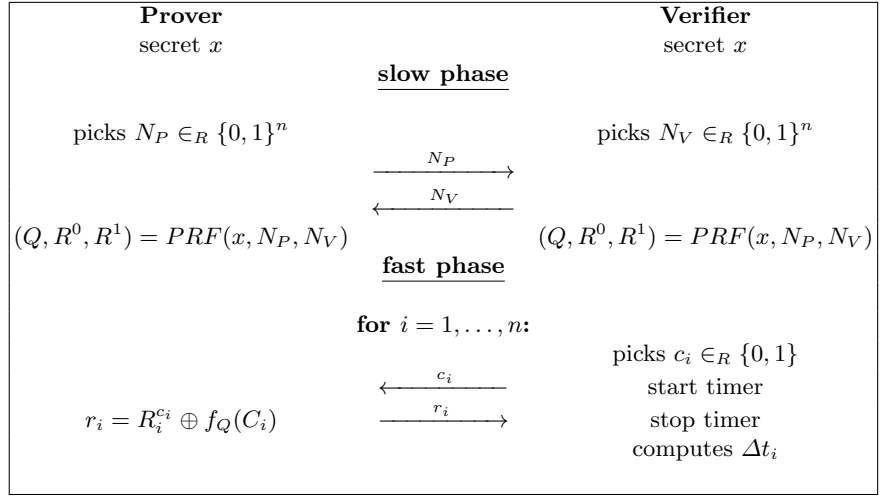

\centering
 {\renewcommand\normalsize{\footnotesize}%
  \normalsize

\begin{tabular}{|ccc|}
\hline
%\phantom{aaaaaaaaaaaaaaaaaaaaaaaaaaaaa}&&\phantom{aaaaaaaaaaaaaaaaaaaaaaaaaaaaa}\\
\normalsize{\textbf{Prover}}&&\normalsize{\textbf{Verifier}}\\
secret $x$&&secret $x$\\
&\textbf{\underline{slow phase}}&\\
&&\\
picks $N_{P} \in_{R} \{0, 1\}^n$ &&picks $N_{V}\in_{R} \{0, 1\}^n$\\
&$\xrightarrow{\eqspace N_{P}\eqspace}$&\\
&$\xleftarrow{\eqspace N_{V}\eqspace}$&\\
%&$\stackrel{N_{P}}{\hbox to 30mm{\rightarrowfill}}$&\\
%&$\stackrel{N_{V}}{\hbox to 30mm{\leftarrowfill}}$&\\

$(Q,R^0,R^1)=PRF(x,N_{P},N_{V})$&&$(Q,R^0,R^1)=PRF(x,N_{P},N_{V})$\\
&\textbf{\underline{fast phase}}&\\
&&\\
&\textbf{for $i=1,\dots,n$:}&\\
&&picks $c_{i} \in_{R} \{0,1\}$\\
&$\xleftarrow{\eqspace c_i\eqspace}$& start timer\\
$r_{i} =R_{i}^{c_i}\oplus f_Q(C_i)$&$\xrightarrow{\eqspace r_i\eqspace}$&stop 
timer\\
 && computes $\Delta t_i$\\
 &&\\\hline
 \end{tabular}}

\caption{Protocol description\label{fig:2b}}
\end{figure*}

\section{Resistance to mafia fraud}\label{sec:mafia}

Analyses of DB protocols usually consider two strategies to evaluate the resistance against a mafia fraud: the pre-ask and the post-ask strategies~\cite{AvoineBKLM-2011-jcs}. Although considering these two strategies only do not provide a formal security proof, this evaluates the resistance of the protocol, at least against these well-known attack strategies, and is the only way known so far to evaluate DB protocols. Providing a formal security proof of DB protocols would be interesting but is clearly out of the scope of this paper.

This section reminds the concept of pre-ask strategy\footnote{Note that the 
post-ask strategy is not relevant in protocols without an extra final slow 
phase~\cite{AvoineBKLM-2011-jcs}}, then identifies the adversarial behavior 
that maximizes the success probability when considering the pre-ask strategy, 
and the section finally computes this probability.

\subsection{Best behavior}

The \emph{pre-ask strategy} consists first, for the adversary, to relays the 
initial slow phase. Then, she runs the fast phase with the prover. With the 
answers she obtains, she finally executes the fast phase with the verifier. In 
our case, we consider that the adversary first sends a sequence of challenges 
$\tilde{c}_1...\tilde{c}_n$ to the legitimate prover and receives 
$\tilde{r}_1...\tilde{r}_n$ where $\tilde{r}_i = R_i^{\tilde{c}_i} \oplus 
f_Q(\tilde{C}_i)$ and $\tilde{C}_i = \tilde{c}_1...\tilde{c}_i$ for every $i 
\in \{1,...,n\}$. Next, she executes the fast phase with the verifier receiving 
the challenges $c_1...c_n$. Given $\tilde{r}_1...\tilde{r}_n$, the adversarial 
behavior that maximizes the success probability is provided in 
Theorem~\ref{theo_best_strategy}.

\begin{theorem}\label{theo_best_strategy}
The adversary's behavior that maximizes her mafia fraud success probability 
with a pre-ask strategy is: (a) For every round $i$ where $c_i \neq 
\tilde{c}_i$, answer randomly. (b) For every round $i$ where $c_i = 
\tilde{c}_i$, guess the value $f_Q(C_i) \oplus f_Q(\tilde{C}_i)$ and answer 
with the value $f_Q(C_i) \oplus f_Q(\tilde{C}_i) \oplus \tilde{r}_i$ where $C_i 
= c_1...c_i$ and $\tilde{C}_i = \tilde{c}_1...\tilde{c}_i$.
\end{theorem}

\begin{proof}

First, let us prove the following lemma.

\begin{lemma}\label{lemma_best_strategy}
Given that $c_i = \tilde{c}_i$, the adversary's behavior maximizing her success 
probability at the $i$th round is equivalent to the best behavior for guessing 
the value $f_Q(C_i) \oplus f_Q(\tilde{C}_i)$.
\end{lemma}

\begin{proof}
Given that $c_i = \tilde{c}_i$, then $R_i^{c_i} = R_i^{\tilde{c}_i}$, which 
means that $\tilde{r}_i \oplus f_Q(\tilde{C}_i) = r_i \oplus f_Q(C_i) 
\Rightarrow r_i =  f_Q(\tilde{C}_i) \oplus f_Q(C_i)\oplus \tilde{r}_i$. 
Therefore, either the adversary guesses $f_Q(C_i) \oplus f_Q(\tilde{C}_i)$ or 
the adversary losses this round.
\end{proof}

In the case where $c_i \neq \tilde{c}_i$, the prover's response $\tilde{r}_i$ 
does not help the adversary since $R_i^{c_i}$ and $R_i^{\tilde{c}_i}$ are 
independent values. Therefore, there does not exist any best behavior, 
\emph{i.e.,} whatever the adversary behavior, her success probability at this 
round is $\frac{1}{2}$. This result and Lemma~\ref{lemma_best_strategy} 
conclude the proof. \qed
\end{proof}

\subsection{Adversary's success probability}

Given the adversary's behavior provided by Theorem~\ref{theo_best_strategy}, Theorem~\ref{theo_mafia} provides a recursive way to compute her success probability.

\begin{theorem}\label{theo_mafia}
Let $M_i$ be the event that the adversary has won the first $i$ rounds by 
following her best behavior with the pre-ask strategy. Let $S_i$ be the event 
that the adversary guesses $f_Q(C_i) \oplus f_Q(\tilde{C}_i)$ at the $i$th 
round. The probability $\Pr(M_i)$ can be recursively computed as follows:

$$
\begin{array}{r l}
\Pr(M_i) = & \left(\frac{1}{2}\right)^i + \Pr(M_i | C_{i} \neq 
\tilde{C}_i)\left(1 - \left(\frac{1}{2}\right)^i \right).\\
\end{array}
$$

$$
\begin{array}{r l}
\Pr(M_i| C_i \neq \tilde{C}_i) = & \Pr(M_i| C_i \neq \tilde{C}_i, M_{i-1})\\
& \times\left(
\frac{1}{2^{i} - 1}  + \Pr(M_{i-1} | C_{i-1} \neq \tilde{C}_{i-1})(1 
-\frac{1}{2^{i} - 1} )\right).\\
\end{array}
$$

$$
\begin{array}{r l}
\Pr(M_i| C_i \neq \tilde{C}_i, M_{i-1}) = & \Pr(S_{i-1}| C_{i-1} \neq 
\tilde{C}_{i-1}, M_{i-1})\left(1 - \frac{2^{i-1}}{2^i-1}\right)+ 
\frac{1}{2}\frac{2^{i-1}}{2^i-1}.\\
\end{array}
$$

$$
\begin{array}{r l}
\Pr(S_{i}| C_{i} \neq \tilde{C}_{i}, M_{i})= & \frac{1}{2} + 
\frac{1}{2}\frac{\Pr(S_{i-1} | M_{i-1}, C_{i-1} \neq 
\tilde{C}_{i-1})\Pr(M_{i-1}| C_{i-1} \neq 
\tilde{C}_{i-1})\left(1-\left(\frac{1}{2} \right)^{i-1} 
\right)\frac{1}{2}}{\Pr(M_{i} | C_{i} \neq 
\tilde{C}_{i})\left(1-\left(\frac{1}{2} \right)^{i} \right)}.\\
\end{array}
$$

\noindent \text{Where } $\Pr(M_{1} | C_{1} \neq \tilde{C}_{1}) = {1}/{2}$ 
\text{ and } $\Pr(S_{1}| C_{1} \neq \tilde{C}_{1}, M_{1}) = {1}/{2}$ \text{ are 
the stopping} \text{ conditions}.
\end{theorem}

\begin{proof}

If $C_i = \tilde{C}_i$, the adversary knows that $f_Q(C_i) \oplus 
f_Q(\tilde{C}_i) = 0$ and thus, her success probability until the $i$th round 
is $1$, which means that $\Pr(M_i | C_{i} = \tilde{C}_{i}) = 1$. Considering 
that $\Pr(C_i = \tilde{C}_i) = \left({1}/{2}\right)^i$ and $\Pr(C_i \neq 
\tilde{C}_i) = 1- \left({1}/{2}\right)^i$, then $\Pr(M_i)$ can be expressed as 
follows:
\begin{equation}\label{reationality_eq_1}
\Pr(M_i) = \left(\frac{1}{2}\right)^i + \Pr(M_i | C_{i} \neq 
\tilde{C}_{i})\left(1 - \left(\frac{1}{2}\right)^i \right).
\end{equation}

\noindent Equation~\ref{reationality_eq_1} states that the computation of 
$\Pr(M_i)$ requires $\Pr(M_i | C_{i} \neq \tilde{C}_{i})$. Note that $M_i$ 
holds if $M_{i-1}$ holds, so:
\begin{align}\label{reationality_eq_1.5}
\Pr(M_i | C_{i} \neq \tilde{C}_{i}) =& \Pr(M_i | C_{i} \neq \tilde{C}_{i}, 
M_{i-1})\Pr(M_{i-1} | C_{i} \neq \tilde{C}_{i}).
\end{align}

\noindent Given that $M_{i-1}$ depends on whether $C_{i-1} = \tilde{C}_{i-1}$ 
or not, and considering that $\Pr(C_{i-1} = \tilde{C}_{i-1} | C_{i} \neq 
\tilde{C}_{i}) = {1}/({2^{i} - 1})$, then $\Pr(M_{i-1} | C_{i} \neq 
\tilde{C}_{i})$ can be transformed as follows:
\begin{align}\label{reationality_eq_2}
&\Pr(M_{i-1}| C_i \neq \tilde{C}_i)  && =  \Pr(M_{i-1} | C_i \neq \tilde{C}_i, 
C_{i-1}= \tilde{C}_{i-1}) \Pr(C_{i-1} = \tilde{C}_{i-1} | C_{i} \neq 
\tilde{C}_{i}) 
\notag\\
& && + \Pr(M_{i-1} | C_i \neq \tilde{C}_i, C_{i-1} \neq 
\tilde{C}_{i-1})\Pr(C_{i-1} \neq \tilde{C}_{i-1} | C_{i} \neq 
\tilde{C}_{i}) 
\notag\\
& && = \frac{1}{2^{i} - 1}  + \Pr(M_{i-1} | C_{i-1} \neq 
\tilde{C}_{i-1})\frac{2^{i} - 2}{2^{i} - 1}.
\end{align}

\noindent Assuming that $\Pr(M_i | C_{i} \neq \tilde{C}_{i}, M_{i-1})$ can be 
computed for every $i$, then according to Equations~\ref{reationality_eq_1.5} 
and~\ref{reationality_eq_2}, $\Pr(M_i | C_{i} \neq \tilde{C}_{i})$ can be 
recursively computed as follows:

\begin{align}\label{reationality_eq_3}
\Pr(M_{i}| C_i \neq \tilde{C}_i) = &\Pr(M_i | C_{i} \neq \tilde{C}_{i}, 
M_{i-1})\notag\\
& \times\left( \frac{1}{2^{i} - 1}  + \Pr(M_{i-1} | C_{i-1} \neq 
\tilde{C}_{i-1})\frac{2^{i} - 2}{2^{i} - 1} \right).
\end{align}

\noindent Note that, the result $\Pr(M_{1} | C_{1} \neq \tilde{C}_{1}) = 
{1}/{2}$ can be used as the stopping condition for the recursion defined in 
Equation~\ref{reationality_eq_3}. This recursion simplifies the analysis of 
$\Pr(M_i)$: instead of analyzing the probability to win all the $i$ rounds, 
only the probability to win the $i$th round is needed. Since it depends on the 
adversary's behavior, and the latter depends on whether $c_i = \tilde{c}_i$ or 
not, we compute  $\Pr(M_i | C_{i} \neq \tilde{C}_{i}, M_{i-1})$ as follows:
\begin{align}\label{reationality_eq_2.5}
&\Pr(M_i| C_i \neq \tilde{C}_i, M_{i-1}) = && \Pr(M_i| C_{i-1} \neq 
\tilde{C}_{i-1}, M_{i-1}, c_i = \tilde{c}_i) 
\Pr(c_i = \tilde{c}_i|C_i \neq \tilde{C}_i) \notag\\
                                       & && +\Pr(M_i| C_i \neq 
                                       \tilde{C}_i, M_{i-1}, c_i \neq 
                                       \tilde{c}_i)\Pr(c_i \neq 
                                       \tilde{c}_i|C_i \neq \tilde{C}_i).
\end{align}

\noindent When $c_i \neq \tilde{c}_i$ the adversary answers randomly and thus 
$\Pr(M_i| C_i \neq \tilde{C}_i, M_{i-1}, c_i \neq \tilde{c}_i) = {1}/{2}$. 
Considering this result and that  $\Pr(c_i \neq \tilde{c}_i|C_i \neq 
\tilde{C}_i) = {2^{i-1}}/({2^i-1})$,  Equation~\ref{reationality_eq_2.5} yields 
to:
\begin{align}\label{reationality_eq_4}
\Pr(M_i| C_i \neq \tilde{C}_i, M_{i-1})= & \frac{2^{i-2}}{2^i-1}\notag\\
& +\Pr(M_i| 
C_{i-1} \neq \tilde{C}_{i-1}, M_{i-1}, c_i = 
\tilde{c}_i)\left( 1-\frac{2^{i-1}}{2^i-1}\right).
\end{align}

%\noindent From Equation~\ref{reationality_eq_4}, we deduce that to compute 
%$\Pr(M_i)$, we need to compute $\Pr(M_i| C_{i-1} \neq \tilde{C}_{i-1}, 
%M_{i-1}, 
%c_i = \tilde{c}_i)$. However, as shown by Theorem~\ref{lemma_mafia_1}, the 
%adversary's strategy in this case is to guess $f_Q(C_i) \oplus 
%f_Q(\tilde{C}_i)$ and her probability to win is the probability of guessing 
%such a value, which means that:

\noindent From Equation~\ref{reationality_eq_4}, we deduce that computing 
$\Pr(M_i)$ requires $\Pr(M_i| C_{i-1} \neq \tilde{C}_{i-1}, M_{i-1}, c_i = 
\tilde{c}_i)$. Theorem~\ref{lemma_mafia_1} states that the adversary's behavior 
in this case is to guess $f_Q(C_i) \oplus f_Q(\tilde{C}_i)$. Hence:
\begin{align}\label{eq_1_corollary_mafia_1}
&\Pr(M_i| C_{i-1} \neq \tilde{C}_{i-1}, M_{i-1}, c_i = \tilde{c}_i) =\Pr(S_i| 
C_{i-1} \neq \tilde{C}_{i-1}, M_{i-1}, c_i = \tilde{c}_i)
\end{align}

\noindent We now aim at computing $\Pr(S_i| C_{i-1} \neq \tilde{C}_{i-1}, 
M_{i-1}, c_i = \tilde{c}_i)$. Since $c_i = \tilde{c}_i$, then  $f_Q(C_i) \oplus 
f_Q(\tilde{C}_i) = f_Q(C_{i-1}) \oplus f_Q(\tilde{C}_{i-1})$. Therefore, the 
adversary's strategy  maximizing $\Pr(S_i| C_{i-1} \neq \tilde{C}_{i-1}, 
M_{i-1}, c_i = \tilde{c}_i)$ consists in holding her previous guess for the 
($i-1$)th round. So:
\begin{align}\label{eq_2_corollary_mafia_1}
&\Pr(S_i| C_{i-1} \neq \tilde{C}_{i-1}, M_{i-1}, c_i = \tilde{c}_i) 
=\Pr(S_{i-1}| C_{i-1} \neq \tilde{C}_{i-1}, M_{i-1}).
\end{align}

\noindent As pointed out by Equation~\ref{eq_2_corollary_mafia_1} and 
Equation~\ref{eq_1_corollary_mafia_1}, computing $\Pr(M_i| C_{i-1} \neq 
\tilde{C}_{i-1}, M_{i-1}, c_i = \tilde{c}_i)$ requires $\Pr(S_{i-1}| C_{i-1} 
\neq \tilde{C}_{i-1}, M_{i-1})$. Since it is indexed by $i-1$, we assume that 
$\Pr(M_j)$ is already computed for every $j < i$ and as shown by 
Lemma~\ref{lemma_mafia_1}, $\Pr(S_{i-1}| C_{i-1} \neq \tilde{C}_{i-1}, 
M_{i-1})$ can be recursively computed.

\begin{lemma}\label{lemma_mafia_1}
Given that $\Pr(M_{j} | C_{j} \neq \tilde{C}_{j})$ can be computed for every $j 
\leq i$, then $\Pr(S_i| C_{i} \neq \tilde{C}_{i}, M_{i})$ can be recursively 
computed as follows:

%\begin{eqnarray*}
%\Pr(S_{i}| C_{i} \neq \tilde{C}_{i}, M_{i}) &=& \frac{1}{2} + 
%\frac{\Pr(S_{i-1} | M_{i-1}, C_{i-1} \neq \tilde{C}_{i-1})}{\Pr(M_{i} | C_{i} 
%\neq \tilde{C}_{i})\left%(1-\left(\frac{1}{2} \right)^{i} \right)}\\
%&&\;\;\;\;\times\Pr(M_{i-1}| C_{i-1} \neq 
%\tilde{C}_{i-1})\left(\frac{1}{2}-\left(\frac{1}{2} \right)^{i} \right)
%\end{eqnarray*}
\begin{align}
\Pr(S_{i}| C_{i} \neq \tilde{C}_{i}, M_{i})=&\frac{1}{2}+ \frac{\Pr(S_{i-1} | 
M_{i-1}, C_{i-1} \neq 
\tilde{C}_{i-1})}{\Pr(M_{i} | C_{i} \neq 
\tilde{C}_{i})\left(1-\left(\frac{1}{2} \right)^{i} \right)}\notag\\
 & \times\Pr(M_{i-1}| C_{i-1} \neq 
\tilde{C}_{i-1})\left(\frac{1}{2}-\left(\frac{1}{2} \right)^{i} \right).\notag
\end{align}
where $\Pr(S_{1}| C_{1} \neq \tilde{C}_{1}, M_{1}) = \frac{1}{2}$ is the 
stopping condition.
\end{lemma}

\begin{proof}

\noindent By definition of $f_Q(.)$, and because $\Pr(q_i = 0) = \Pr(q_i = 1) = 
{1}/{2}$, $\Pr(S_{i}| C_{i} \neq \tilde{C}_{i}, M_{i}, c_i \neq \tilde{c}_i) = 
{1}/{2}$. Moreover, Theorem~\ref{theo_best_strategy} states that $M_i$ and 
$S_i$ are equivalent when $c_i = \tilde{c}_i$, hence $\Pr(S_{i}| C_{i} \neq 
\tilde{C}_{i}, M_{i}, c_i = \tilde{c}_i) = 1$. Considering both results we 
obtain:

\begin{align}\label{eq_1_mafia_1}
&\Pr(S_{i}| C_{i} \neq \tilde{C}_{i}, M_{i}) = \frac{1}{2} + 
\frac{1}{2}\Pr(c_{i} = \tilde{c}_{i}|C_{i} \neq \tilde{C}_{i}, 
M_{i}).
\end{align}

%\begin{align}\label{eq_1_mafia_1}
%\Pr(S_{i}| C_{i} \neq \tilde{C}_{i}, M_{i}) = & \Pr(S_{i}| C_{i} \neq 
%\tilde{C}_{i}, M_{i}, c_i \neq \tilde{c}_i)\Pr(c_i \neq \tilde{c}_i|C_{i} \neq 
%\tilde{C}_{i}, M_{i})\notag\\
%+ & \Pr(S_{i}| C_{i} \neq \tilde{C}_{i}, M_{i}, c_i = \tilde{c}_i)\Pr(c_i = 
%\tilde{c}_i|C_{i} \neq \tilde{C}_{i}, M_{i})\notag\\
%%= & \frac{1}{2}\Pr(c_i \neq \tilde{c}_i|C_{i} \neq \tilde{C}_{i}, 
%%M_{i})+\Pr(c_i = \tilde{c}_i|C_{i} \neq \tilde{C}_{i}, M_{i})\notag\\
%= & \frac{1}{2} + \frac{1}{2}\Pr(c_{i} = \tilde{c}_{i}|C_{i} \neq 
%\tilde{C}_{i}, M_{i}).
%\end{align}

\noindent Finally, considering that $\Pr(C_{i} \neq \tilde{C}_{i}) = 
1-\left(\frac{1}{2} \right)^{i} $ and $\Pr(c_i = \tilde{c}_i, C_{i} \neq 
\tilde{C}_{i}) = \left(1-\left(\frac{1}{2} \right)^{i-1} \right)\frac{1}{2}$, 
the probability $\Pr(c_i = \tilde{c}_i|C_{i} \neq \tilde{C}_{i}, M_{i})$ can be 
expressed as follows:
\begin{align}\label{eq_2_mafia_1}
\Pr(c_i = \tilde{c}_i|C_{i} \neq \tilde{C}_{i}, M_{i})  = & \frac{\Pr(c_i = 
\tilde{c}_i, C_{i} \neq \tilde{C}_{i}, M_{i}, 
M_{i-1})}{\Pr(C_{i} \neq \tilde{C}_{i}, M_{i})} \notag\\
= & \frac{\Pr(M_i | M_{i-1}, c_i = \tilde{c}_i, C_{i} \neq 
\tilde{C}_{i})}{\Pr(M_{i} | C_{i} \neq \tilde{C}_{i})\Pr(C_{i} \neq 
\tilde{C}_{i})} \notag\\
& \times\Pr(M_{i-1}| c_i = \tilde{c}_i, C_{i} \neq 
\tilde{C}_{i})\Pr(c_i = \tilde{c}_i, C_{i} \neq \tilde{C}_{i})\notag\\
= & \frac{\Pr(S_{i-1} | M_{i-1}, C_{i-1} \neq \tilde{C}_{i-1})}{\Pr(M_{i} | 
C_{i} 
\neq \tilde{C}_{i})\left(1-\left(\frac{1}{2} \right)^{i} 
\right)}\\\notag
& \times\Pr(M_{i-1}| C_{i-1} \neq 
\tilde{C}_{i-1})\left(\frac{1}{2}-\left(\frac{1}{2} \right)^{i} \right).\notag
\end{align}
\noindent Equations~~\ref{eq_1_mafia_1} and~\ref{eq_2_mafia_1} yield the 
expected result. 

\end{proof}

\noindent Lemma~\ref{lemma_mafia_1} together with  
Equations~\ref{reationality_eq_1},~\ref{reationality_eq_3},~\ref{reationality_eq_4},~\ref{eq_1_corollary_mafia_1},
 and~\ref{eq_2_corollary_mafia_1}, conclude the proof of this theorem.

\end{proof}

%\subsection{Distance Fraud.}
\section{Resistance to distance fraud} \label{sec:distance}

This section analyzes the adversary success probability when mounting a 
distance fraud. As stated in~\cite{AvoineBKLM-2011-jcs}, the common way to 
analyze the resistance to distance fraud is by considering the early-reply 
strategy instead of the pre-ask strategy. This strategy consists on sending the 
responses in advance, \emph{i.e.,} before receiving the challenges. Doing so, 
the adversary gains some time and might pass the timing constraint. In this 
section, the behavior that maximizes the success probability using the 
early-reply strategy is identified, and then a recursive way to compute the 
resistance w.r.t. a distance fraud is provided.

\subsection{Best behavior}

With the early-reply strategy, in order to send a response in advance in the 
$i$th round with probability of being correct greater that ${1}/{2}$, the 
adversary must send either $R_i^{0} \oplus f_Q({C}_{i-1}||0)$ or $R_i^{1} 
\oplus f_Q({C}_{i-1}||1)$ where ${C}_{i-1}$ is the sequence of challenges sent 
by the verifier until the $(i-1)$th round. Theorem~\ref{theo_strategy_distance} 
shows that, guessing the values $f_Q({C}_{i})$ for every $i \in \{1, ..., 
n-1\}$ is needed to maximize the adversary success probability.

\begin{theorem}\label{theo_strategy_distance}
Let ${C}_{i}$ be the sequence of challenges $c_1...c_i$ sent by the verifier 
until the $i$th round ($i \geq 1$). The adversary's behavior that maximizes her 
distance 
fraud success probability is equivalent to the best behavior for guessing the 
values \\$f_Q(C_1), f_Q(C_2), ..., f_Q(C_{n-1})$.
\end{theorem}

\begin{proof}

\noindent In order to send a response in advance at the $i$th round with 
probability of being correct greater that ${1}/{2}$, the adversary must send 
either $R_i^{0} \oplus f_Q({C}_{i-1}||0)$ or $R_i^{1} \oplus 
f_Q({C}_{i-1}||1)$. By definition, $f_Q({C}_{i-1}||0) = f({C}_{i-1})$ and 
$f_Q({C}_{i-1}||1) = f_Q({C}_{i-1}) \oplus q_i$. Therefore, $R_i^{0} \oplus 
f_Q({C}_{i-1}||0) = R_i^{0} \oplus f_Q({C}_{i-1})$ and $R_i^{1} \oplus 
f_Q({C}_{i-1}||1) = R_i^{1} \oplus f_Q({C}_{i-1}) \oplus q_i$. Since the 
adversary knows the values $R_i^{0}, R_i^{1}$, and $q_i$, guessing the correct 
value at this round is equivalent to guessing the correct value of 
$f_Q({C}_{i-1})$. 

\end{proof}

As stated in Theorem~\ref{theo_strategy_distance}, computing the adversary 
success probability requires the best behavior to guess the outputs sequence 
$f_Q(C_1),\dots,f_Q(C_{n-1})$. Theorem~\ref{theo_distance_f} solves this 
problem.

\begin{theorem}\label{theo_distance_f}
The best adversary's behavior to guess $f_Q(C_{i})$ is to assume that her 
previous guess for $f_Q(C_{i-1})$ is correct and to compute $f_Q(C_{i})$ as 
follows: (a) if $q_i = 0$, then consider $f_Q(C_{i}) =f_Q(C_{i-1})$. (b) if 
$q_i = 1$, pick a random bit $c_i$ and consider that $f_Q(C_{i}) =f_Q(C_{i-1}) 
\oplus c_i$.

\end{theorem}

\begin{proof}

\noindent Assuming that $q_i = 0$, then $f_Q(C_{i}) = f_Q(C_{i-1})$ and thus, 
the probability to guess $f_Q(C_{i})$ is equal to the probability of guessing 
$f_Q(C_{i-1})$. In the case of $q_i = 1$, the adversary does not have a better 
behavior than choosing a random bit of challenge $c_i$ and considering that 
$f_Q(C_i) = f_Q(C_{i-1}) \oplus c_i$. Given that $f_Q(\epsilon) = 0$ where 
$\epsilon$ is the empty sequence, the proof can be straightforwardly concluded 
by induction. 

\end{proof}

\subsection{Adversary's success probability}

Given the best adversary's behavior provided by Theorems~\ref{theo_strategy_distance} and~\ref{theo_distance_f}, Theorem~\ref{theo_distance} shows a recursive way to compute the resistance to distance fraud.

\begin{theorem}\label{theo_distance}
Let $D_i$ be the event that the distance fraud adversary successfully passes 
the protocol until the $i$th round. Then, $\Pr(D_i)$ can be computed as follows:

\begin{equation*}
\Pr(D_i) = \frac{1}{4}\Pr(D_{i-1}) + \frac{1}{2^i} + \frac{1}{8}\sum_{j = 
1}^{i-1} \Pr(D_{j})\frac{1}{2^{i-j}}.	
\end{equation*}
where $\Pr(D_0) = 1$ is the stopping condition.
\end{theorem}

\begin{proof}

Let $F_i$ be the event that the adversary correctly guesses the value of 
$f_Q(C_i)$. Then, the event $D_i$ depends on the events $D_{i-1}$ and 
$F_{i-1}$, which can be expressed as follows:
\begin{align}\label{theo_distance_eq_1}
\Pr(D_i) =& \Pr(D_i| D_{i-1}, F_{i-1})\Pr(D_{i-1}, F_{i-1})\notag\\ &\eqspace+ 
\Pr(D_i| D_{i-1}, \bar{F}_{i-1})\Pr(D_{i-1}, \bar{F}_{i-1}). 
\end{align}

\noindent Two cases occur (a) $R_i^0 = R_i^1$ and (b) $R_i^0 \neq R_i^1$. In 
the first case, the adversary wins the $i$th round if and only if she guesses 
the value $f_Q(C_{i-1})$, so $\Pr(D_i| D_{i-1}, F_{i-1}, R_i^0 = R_i^1) = 1$ 
and $\Pr(D_i| D_{i-1}, \bar{F}_{i-1}, R_i^0 = R_i^1) = 0$. In the second case, 
the adversary has no better probability to win than ${1}/{2}$ and thus, 
$\Pr(D_i| D_{i-1}, F_{i-1}, R_i^0 \neq R_i^1) = \Pr(D_i| D_{i-1}, 
\bar{F}_{i-1}, R_i^0 \neq R_i^1) = {1}/{2}$. Therefore, we deduce $\Pr(D_i| 
D_{i-1}, F_{i-1}) = {3}/{4}$ and $\Pr(D_i| D_{i-1}, \bar{F}_{i-1}) = {1}/{4}$. 
Using these results and  Equation~\ref{theo_distance_eq_1} we have:
\begin{align}\label{theo_distance_eq_2}
\Pr(D_i) &=\frac{3}{4}\Pr(D_{i-1}, F_{i-1}) + \frac{1}{4}\Pr(D_{i-1}, 
\bar{F}_{i-1}) \notag \\
&= \frac{1}{4}\Pr(D_{i-1}) + \frac{1}{2}\Pr(D_{i-1}, F_{i-1}).
\end{align}

\noindent Equation~\ref{theo_distance_eq_2} states that $\Pr(D_i)$ can be 
computed by recursion if we express $\Pr(D_{i-1}, F_{i-1})$ in terms of the 
events $D_j$ where $1 \leq j < i$. Therefore, in the remaining of this proof we 
aim at looking for such result. As above, in order to analyze $\Pr(D_{i}, 
F_{i})$, the events $D_{i-1}$ and $F_{i-1}$ should be considered:

\begin{align}\label{theo_distance_eq_3}
\Pr(D_{i}, F_{i})  = & \Pr(D_{i}| F_{i}, D_{i-1}, F_{i-1})\Pr(F_{i}| 
D_{i-1}, F_{i-1})\Pr(D_{i-1}, F_{i-1})\notag 
\\
&+\Pr(D_{i}| F_{i}, D_{i-1}, \bar{F}_{i-1})\Pr(F_{i}| D_{i-1}, 
\bar{F}_{i-1})\Pr(D_{i-1}, \bar{F}_{i-1}).
\end{align}

\noindent Four cases should be analyzed depending on the value of $q_i$ and the 
events $F_i$ and $F_{i-1}$. 

\begin{case}[$q_i = 1$, $F_{i}$ and $F_{i-1}$ hold]\label{case1}
This case occurs if the adversary correctly guesses the challenge $c_i$. 
Therefore, she provides the correct answer at this round $R_i^{c_i} \oplus 
f_Q(C_i)$. So, $\Pr(D_{i}| F_{i}, D_{i-1}, F_{i-1}, q_i = 1) = 1$.
\end{case}

\begin{case}[$q_i = 1$, $F_{i}$ and $\bar{F}_{i-1}$ hold]\label{case2}
Given that $\bar{F}_{i-1}$ and $F_{i}$ hold, the adversary computed $f_Q(C_i) = 
f_Q(C_{i-1}) \oplus \tilde{c}_i$ using a challenge different from the 
verifier's one, \emph{i.e.,} $c_i \neq \tilde{c}_i$. Therefore, $\Pr(D_{i}| 
F_{i}, D_{i-1}, \bar{F}_{i-1}, q_i = 1) = \frac{1}{2}$ because both events 
$F_{i}$ and $\bar{F}_{i-1}$ coexist only if $q_i = 1$, then $\Pr(q_i = 1|F_{i}, 
D_{i-1}, \bar{F}_{i-1}) = 1$.
\end{case}

\begin{case}[$q_i = 0$, $F_{i}$ and $F_{i-1}$ hold]\label{case3}
Given $q_i = 0$, the event $F_i$ has no effect on the event $D_i$. Thus, 
$\Pr(D_{i}| F_{i}, D_{i-1}, F_{i-1}, q_i = 0) = \Pr(D_{i}| D_{i-1}, F_{i-1}, 
q_i = 0)=\frac{3}{4}$ because it depends on whether $R_i^0 = R_i^1$. So, 
$\Pr(D_{i}| F_{i}, D_{i-1}, F_{i-1}, q_i = 0) = \frac{3}{4}$.
\end{case}

\begin{case}[$q_i = 0$, $F_{i}$ and $\bar{F}_{i-1}$ hold]\label{case4}
When $q_i = 0$, then $f_Q(C_i) = f_Q(C_{i-1})$, which means that this case 
cannot occur. Therefore, $\Pr(F_{i}, D_{i-1}, \bar{F}_{i-1}, q_i = 0) = 0$.
\end{case}

%\vspace*{-20pt}
%\hspace*{-4pt} 

Cases~\ref{case1} and~\ref{case3} yield  the following result: 

\begin{align}\label{theo_distance_eq_4}
\Pr(D_{i}| F_{i}, D_{i-1}, F_{i-1})= & \Pr(q_i = 1|F_{i}, D_{i-1}, F_{i-1}) 
+ \frac{3}{4}\Pr(q_i = 0|F_{i}, D_{i-1}, F_{i-1}) \notag\\
= &\frac{3}{4} - \frac{1}{4}\Pr(q_i = 1|F_{i}, D_{i-1}, F_{i-1}).
\end{align}

And Cases~\ref{case2} and~\ref{case4} yield this other result: 
\begin{equation}\label{theo_distance_eq_5}
\Pr(D_{i}| F_{i}, D_{i-1}, \bar{F}_{i-1}) = \frac{1}{2}.
\end{equation}

\noindent Because $\Pr(F_{i}| F_{i-1}, q_i = 0) = 1$ and $\Pr(F_{i}| F_{i-1}, 
q_i = 1) = {1}/{2}$, we have $\Pr(F_{i}| D_{i-1}, F_{i-1}) = \Pr(F_{i}| 
F_{i-1}) = {3}/{4}$. Similarly, $\Pr(F_{i}| D_{i-1}, \bar{F}_{i-1}) = 
\Pr(F_{i}| \bar{F}_{i-1}) = {1}/{4}$ because $\Pr(F_{i}| \bar{F}_{i-1}, q_i = 
0) = 0$ and $\Pr(F_{i}| \bar{F}_{i-1}, q_i = 1) = {1}/{2}$. Combining these 
results with Equations~\ref{theo_distance_eq_4} and~\ref{theo_distance_eq_5}, 
Equation~\ref{theo_distance_eq_3} becomes:
\begin{align}\label{theo_distance_eq_6}
\Pr(D_{i}, F_{i})  = & \left(\frac{3}{4} - \frac{1}{4}\Pr(q_i = 1|F_{i}, 
D_{i-1}, F_{i-1})\right)\frac{3}{4}\Pr(D_{i-1}, 
F_{i-1})\notag\\
& +\frac{1}{2}\frac{1}{4}\Pr(D_{i-1}, \bar{F}_{i-1})\notag \\
= &  \frac{3}{16}\Pr(q_i = 1|F_{i}, D_{i-1}, F_{i-1})\Pr(D_{i-1}, F_{i-1})+ 
\frac{9}{16}\Pr(D_{i-1}, F_{i-1})\notag \\
&+\frac{1}{8}\Pr(D_{i-1}, \bar{F}_{i-1})\notag\\
 = & \frac{3}{16}\frac{\Pr(q_i = 1, F_{i}, D_{i-1}, F_{i-1})}{\Pr(F_{i}| 
D_{i-1}, F_{i-1})}+\frac{9}{16}\Pr(D_{i-1}, F_{i-1})+\frac{1}{8}\Pr(D_{i-1}, 
\bar{F}_{i-1})\notag\\
 = & \frac{3}{16}\frac{\Pr(F_{i}|q_i = 1, D_{i-1}, F_{i-1})\Pr(D_{i-1}, 
F_{i-1})\frac{1}{2}}{\Pr(F_{i}| D_{i-1}, F_{i-1})}+\frac{9}{16}\Pr(D_{i-1}, 
F_{i-1})\notag \\
& + \frac{1}{8}\Pr(D_{i-1}, \bar{F}_{i-1})\notag\\
 = &\frac{3}{16}\frac{\frac{1}{2}\Pr(D_{i-1}, 
F_{i-1})\frac{1}{2}}{\frac{3}{4}}+ \frac{9}{16}\Pr(D_{i-1}, 
F_{i-1})+\frac{1}{8}\Pr(D_{i-1}, \bar{F}_{i-1})\notag\\
 = & \frac{1}{16}\Pr(D_{i-1}, F_{i-1})+\frac{9}{16}\Pr(D_{i-1}, 
F_{i-1})+\frac{1}{8}\Pr(D_{i-1}, \bar{F}_{i-1})\notag\\
 = & \frac{5}{8}\Pr(D_{i-1}, F_{i-1}) + \frac{1}{8}(\Pr(D_i) - \Pr(D_{i-1}, 
F_{i-1}))\notag\\
 = &\frac{1}{2}\Pr(D_{i-1}, F_{i-1}) + \frac{1}{8}\Pr(D_i)\notag\\
 = & \frac{1}{2^i} + \frac{1}{8}\sum_{j = 1}^{i} \Pr(D_{j})\frac{1}{2^{i-j}}.
\end{align}

\noindent Considering that $\Pr(D_0) = 1$, Equations~\ref{theo_distance_eq_6} 
and~\ref{theo_distance_eq_2} yield the expected result. 

\end{proof}

\section{Noise resilience}\label{section:noise}

Some efforts have been made in order to adapt existing distance-bounding 
protocols to noisy channels. Most of them rely on using a threshold $x$ 
representing the maximum number of incorrect responses expected by the 
verifier~\cite{Hancke:2005:RDB:1128018.1128472,KimA-2011-ieeetwc}.
 Intuitively, the larger $x$, the lower the false rejection ratio but also the 
lower the resistance to mafia and distance frauds. Others use an error 
correction code during an extra slow 
phase~\cite{Singelee:2007:DBN:1784404.1784415}. However, the latter cannot be 
applied to our protocol given that it does not contain any final slow phase. 
Consequently, 
we focus on the threshold technique. 

\subsection{Understanding the noise effect in our protocol}

We consider in the analysis that a 1-bit challenge (on the \emph{forward} 
channel) can be flipped due to noise with probability $p_f$ and a 1-bit answer 
(on the \emph{backward} channel) can be flipped with probability $p_b$.  
Further, we denote as $\tilde{c}_i$ the bit-challenge received by the prover at 
the $i$th round, which might be obviously different to the challenge $c_i$. 
Similarly, $\tilde{r}_i$ denotes the response received by the verifier at the 
$i$th round. As in previous works~\cite{Hancke:2005:RDB:1128018.1128472}, the 
considered forward and backward channels are assumed to be memoryless. 
Table~\ref{table_noise} shows the three different scenarios when considering a 
noisy communication at the $i$-th round in our protocol.

\begin{table*}
%\vspace{-8pt}
\centering
 {\renewcommand\normalsize{\footnotesize}%
  \normalsize
\begin{tabular}{|l|l|l|l|} \hline
\textbf{Forward Noise} & \textbf{Backward Noise} & \textbf{Forward and}\\
&&\textbf{Backward Noise} \\
\hline
$P$ receives $\tilde{c}_i = c_i \oplus 1$ & $P$ receives $\tilde{c}_i = c_i$  & 
$P$ receives $\tilde{c}_i = c_i \oplus 1$ \\
$P$ updates $\tilde{C}_i = \tilde{c}_1...\tilde{c}_i$&$P$ updates $\tilde{C}_i 
= \tilde{c}_1...\tilde{c}_i$& $P$ updates $\tilde{C}_i = 
\tilde{c}_1...\tilde{c}_i$\\
%$P$ computes $f_Q(\tilde{C}_i)$&$T$ computes $f_Q(\tilde{C}_i)$& $T$ computes 
%$f_Q(\tilde{C}_i)$\\
$P$ sends $r_i = R_i^{\tilde{c}_i} \oplus f_Q(\tilde{C}_i)$&$P$ sends 
$R_i^{c_i} \oplus f_Q(\tilde{C}_i)$& $P$ sends $R_i^{\tilde{c}_i} \oplus 
f_Q(\tilde{C}_i)$\\
$V$ receives $\tilde{r}_i = r_i$&$V$ receives $\tilde{r}_i = r_i \oplus 1$& $V$ 
receives $\tilde{r}_i = r_i \oplus 1$\\\hline
\end{tabular}}
\vspace{10pt}
\caption{The three possible scenarios when some noise occurs at the $i$th 
round.}
\label{table_noise}
%\vspace{-15pt}
\end{table*}

According to the protocol, in a noise-free $i$th round 
executed with a legitimate prover it holds that $r_i = \tilde{r}_i 
\Leftrightarrow f_Q(C_{i}) = f_Q(\tilde{C}_{i})$. We thus say that prover and 
verifier are \emph{synchronized} at the $i$th round if $f_Q(C_{i}) = 
f_Q(\tilde{C}_{i})$, otherwise they are said to be \emph{desynchronized}. 
Intuitively, in a noise-free $i$th round the answer $\tilde{r}_i$ can be 
considered correct by the verifier either if $r_i = \tilde{r}_i$ and they are 
synchronized or if $r_i \neq \tilde{r}_i$ and they are desynchronized. 

The challenge is therefore to identify whether the prover and the verifier are 
synchronized or not. To that aim, we rise the following observation.

\begin{observation}\label{obs_sync}
Several consecutive rounds where all, or almost all, the answers hold that $r_i 
= \tilde{r}_i$ (resp. $r_i \neq \tilde{r}_i$), might indicate that the 
legitimate prover and the verifier have been synchronized (resp. 
desynchronized). 
\end{observation}

%Note that, Observation~\ref{obs_sync} is based on the practical assumption 
%that noise occurs ``unfrequently'', otherwise the data corruption 

\subsection{Our noise resilient mechanism}\label{sec:noise}

Based on Observation~\ref{obs_sync}, we propose an heuristic aimed at 
identifying those rounds where prover and verifier \emph{switch} from being 
synchronized to desynchronized or vice versa. The heuristic is named 
\emph{SwitchedRounds} and its pseudocode description is provided by 
Algorithm~\ref{alg_des}.

\emph{SwitchedRounds} creates first the sequence $d_1...d_n$ where $d_i = 0$ 
if $r_i = \tilde{r}_i$, otherwise $d_i = 1$. Following 
Observation~\ref{obs_sync}, it searches for the longest subsequence 
$d_i...d_{j}$ that matches any of the following patterns\footnote{The patterns 
have been written 
following the POSIX Extended Regular Expressions standard. The symbols $\wedge$ 
and $\$$ represent the start and the 
end of the string respectively.}: (a) $\wedge(1+)0$ (b) $\wedge(1+)\$$ (c) 
$1(0+)1$ (d) $0(1+)0$ (e) $1(0+)\$$ (f) $0(1+)\$$.

The aim of these patterns is to look for \emph{large} subsequences of either 
consecutive $0$s or $1$s in $d_1 \dots d_n$. Note that, we do not include the 
patterns $\wedge(0+)1$ and $\wedge(0+)\$$ because starting with a sequence of 
zeros is exactly what the verifier expects. Intuitively, the lower the 
expected noise the larger 
the subsequences should be. As an example, let us consider the case where the 
communication channel is noiseless. Since no noise is expected, the sequence 
$d_1 \dots d_n$ should be equal to $n$ consecutive zeros unless 
an attack is being performed. In Algorithm~\ref{alg_des}, a threshold $\Delta 
l$ defines how large a matching 
$d_i...d_{j}$  should be in order to be analyzed. We discuss a computational 
approach to estimate a practical value for $\Delta 
l$ in Section~\ref{section:result}.

If a pattern $d_i...d_{j}$ holds that $j - i \geq 
\Delta l$, \emph{SwitchedRounds} looks for the closer index $r$ to $i+1$ such 
that $q_r = 1$, and assumes that the $r$th round caused the switch from  
synchronization to desynchronization or vice versa. To understand this choice, 
let us note that a pattern $d_i...d_{j}$ implies that $d_i \neq d_{i+1}$. 
This could have happened if a switch from synchronization to desynchronization 
or vice versa occurred in the $(i+1)$th round. However, due to the 
probabilistic nature of the noise we cannot precisely determine whether the 
switch occurred in the 
$(i+1)$th round or in some (possibly close) round. What we do know is that such 
a round $r$ must hold that $q_r = 1$, which justifies Step~\ref{step_index_r} 
in 
Algorithm~\ref{alg_des}.

Once $r$ is found, the pair 
$(r,s)$  is created where $s$ is $0$ if the switch is to synchronization, $s = 
1$ otherwise. Finally, \emph{SwitchedRounds} recursively calls itself to 
analyze the two subsequences lying on the left and right side of 
$d_i...d_{j}$. The output is the union of all obtained pairs such that they 
are in increasing order (according to the round) and every two consecutive 
pairs have different values (according to the type of switching). 

\begin{algorithm}[htb]
\caption{SwitchedRounds} \label{alg_des}
\begin{algorithmic}[1]
\Require The challenges $c_1...c_n$ and the registers $R^0$, $R^1$, and $Q$. 
The prover's responses $\tilde{r}_1...\tilde{r}_n$. A threshold $\Delta l$ 
indicating the minimum matching length.
%Two thresholds $min$ and $max$ indicating the indexes that should be 
%considered. 
\State Let $d_1...d_n$ be a sequence such that $d_i = \tilde{r}_i \oplus 
R_i^{c_i}\oplus f(C_i)$.
\State Let $d_i...d_{j}$ be the longest matching with $\wedge(1+)0 | 
\wedge(1+)\$ | 0(1+)0 | 0(1+)\$ | 1(0+)1 | 1(0+)\$$ on $d_{1}...d_{n}$.
\If{$j - i < \Delta l$ or no matching exists}
	\Return the empty set;
\EndIf
\State Let $s$ be $0$ if the matching is with $1(0+)1 | 1(0+)\$$ and $1$ 
otherwise;
\State Let $r$ be the closest index to $i+1$ such that $q_r = 
1$;\label{step_index_r}
\State Let $A$ be the output of \emph{SwitchedRounds} on $d_1...d_{i-1}$;
\State Let $B$ be the output of \emph{SwitchedRounds} on $d_{j+1}...d_{n}$;
\State Let $E$ be the union of $A \cup \{(r, s)\} \cup B$ such that the indexes 
are in increasing order and every two consecutive pairs have different boolean 
values;
\State \Return $E$;
\end{algorithmic}
\end{algorithm}

Armed with the \emph{SwitchedRounds} algorithm, the threshold technique can be 
straightforwardly applied as Algorithm~\ref{alg_authentication} shows. It 
simply counts the number of errors occurred during the execution of the 
protocol where an error is defined as either a switched round or a wrong 
response. Both cases are considered as \emph{error} because, on the one hand,  
a switched round might be falsely detected during an attack,  and on the other 
hand, there is no distinction between a wrong response due to noise or to an 
attack. Finally, the protocol is considered to fail if the number of errors is 
above a threshold $x$.

Note that basing the decision on a threshold is a common and easy
procedure but not the best one, especially when the channels are not
memoryless. Instead, the decision procedure could consist in comparing
the vector $d_1...d_n$ with the error distribution on the noisy channels.

\begin{algorithm}
\caption{Authentication in the presence of noise} \label{alg_authentication}
\begin{algorithmic}[1]
\Require All the parameters of the protocol; an integer value $x$ representing 
the noise tolerance; and a threshold $\Delta l$.
\State Let $E$ be the output of \emph{SwitchedRounds} algorithm on input 
$c_1...c_n$, $\tilde{r}_1...\tilde{r}_n$, $R^0$, $R^1$, $Q$, and $\Delta l$;
\State Let $d_1...d_n$ be a sequence such that $d_i = \tilde{r}_i \oplus 
R_i^{c_i}\oplus f(C_i)$.
\State Let $s$ be a boolean variable initialized in $0$;
\State Let $errors$ be a counter initialized in $0$;
\ForAll{$1 \leq i \leq n$}
%TRUJILLO: I change this a little, please check it.
	\If{$(i, s') \in E$} 
		\State $s \leftarrow s'$; 
		\State $errors++$;
	%\State $s \leftarrow s \oplus 1$; 
	\ElsIf{($d_i = 0$ and $s = 1$) or ($d_i = 1$ and $s = 0$)} 
		\State $errors++$;
		%\State $errors++$;
	\EndIf
\EndFor
\If{$errors > x$}
	\Return fail;
\Else \ 
	\Return success;
\EndIf
\end{algorithmic}
\end{algorithm}

\section{Experiments and comparison}\label{section:result}

The first part of this section is devoted to compare several DB protocols in 
term of mafia fraud resistance, distance fraud resistance, and memory 
consumption. The second part takes noise into account and evaluates our 
proposal w.r.t. the Hancke and Kuhn's~\cite{Hancke:2005:RDB:1128018.1128472} 
and Kim and Avoine's~\cite{KimA-2011-ieeetwc} protocols.

\subsection{Noise-free environment}

Mafia and distance fraud analyses in a noise-free environment can be found 
in~\cite{Hancke:2005:RDB:1128018.1128472,KimA-2011-ieeetwc,Trujillo-Rasua:2010:PDP:1926325.1926352,AvoineT-2009-isc}
 for KA2, AT, Poulidor, and HK. In the case of AT and Poulidor, only an 
upper-bound of their resistance to distance fraud have been 
provided~\cite{Trujillo-Rasua:2010:PDP:1926325.1926352,T2013}. Considering 
those previous results, Fig.~\ref{fig:mafia} and Fig.~\ref{fig:distance} show 
the resistance to mafia and distance frauds respectively for the five 
considered 
protocols in a single chart. For each of them, the configuration that maximizes 
its security has been chosen: this is particularly important for AT and KA2 
because different configurations can be used. 

\begin{figure*}
%\vspace{-10pt}
\centering
  \subfigure[Mafia Fraud]{\hspace{-5pt}
       \includegraphics[width=0.8\textwidth, angle=270]{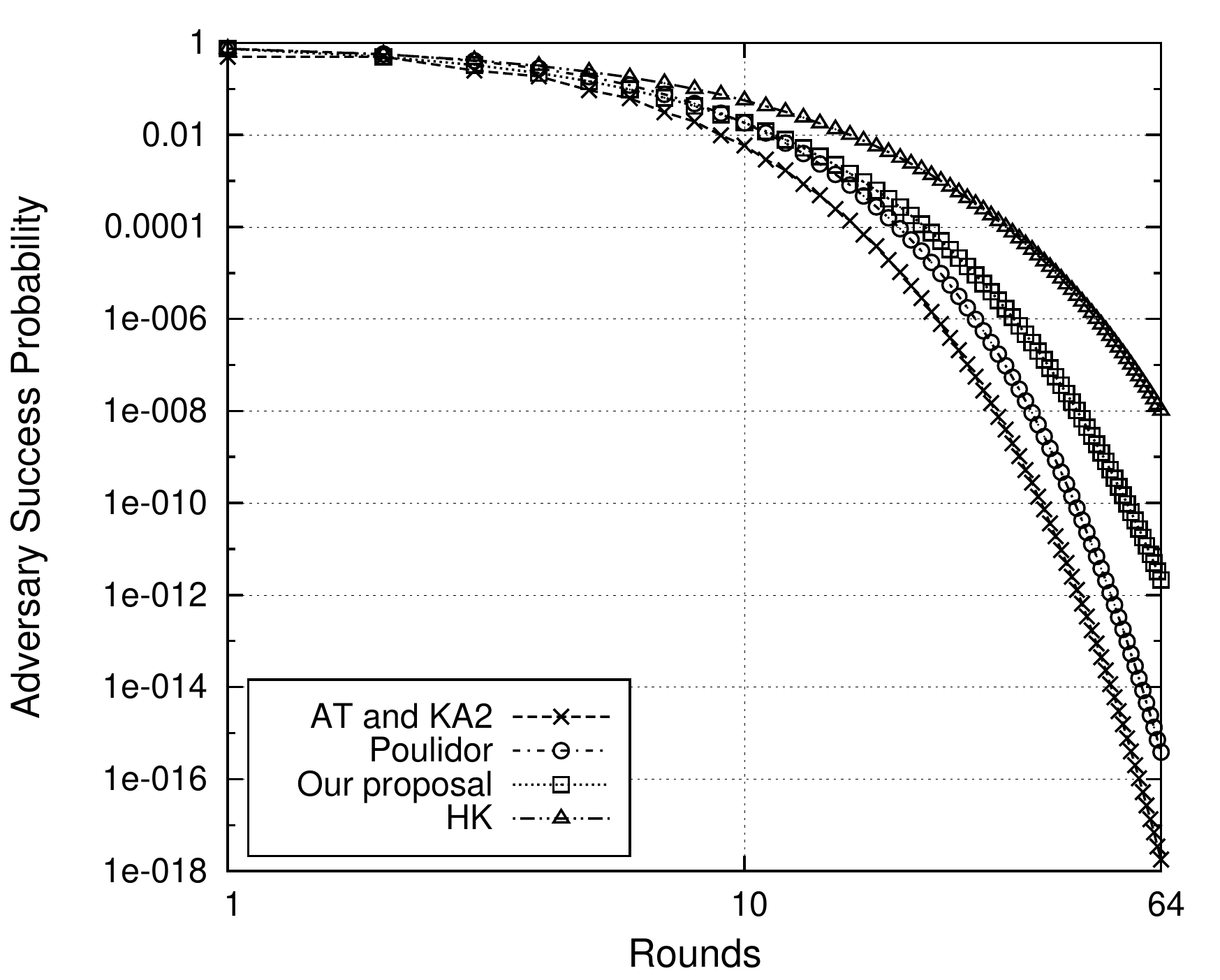}
  	\label{fig:mafia}
  }
  \subfigure[Distance Fraud]{\hspace{-15pt}
       \includegraphics[width=0.8\textwidth, angle=270]{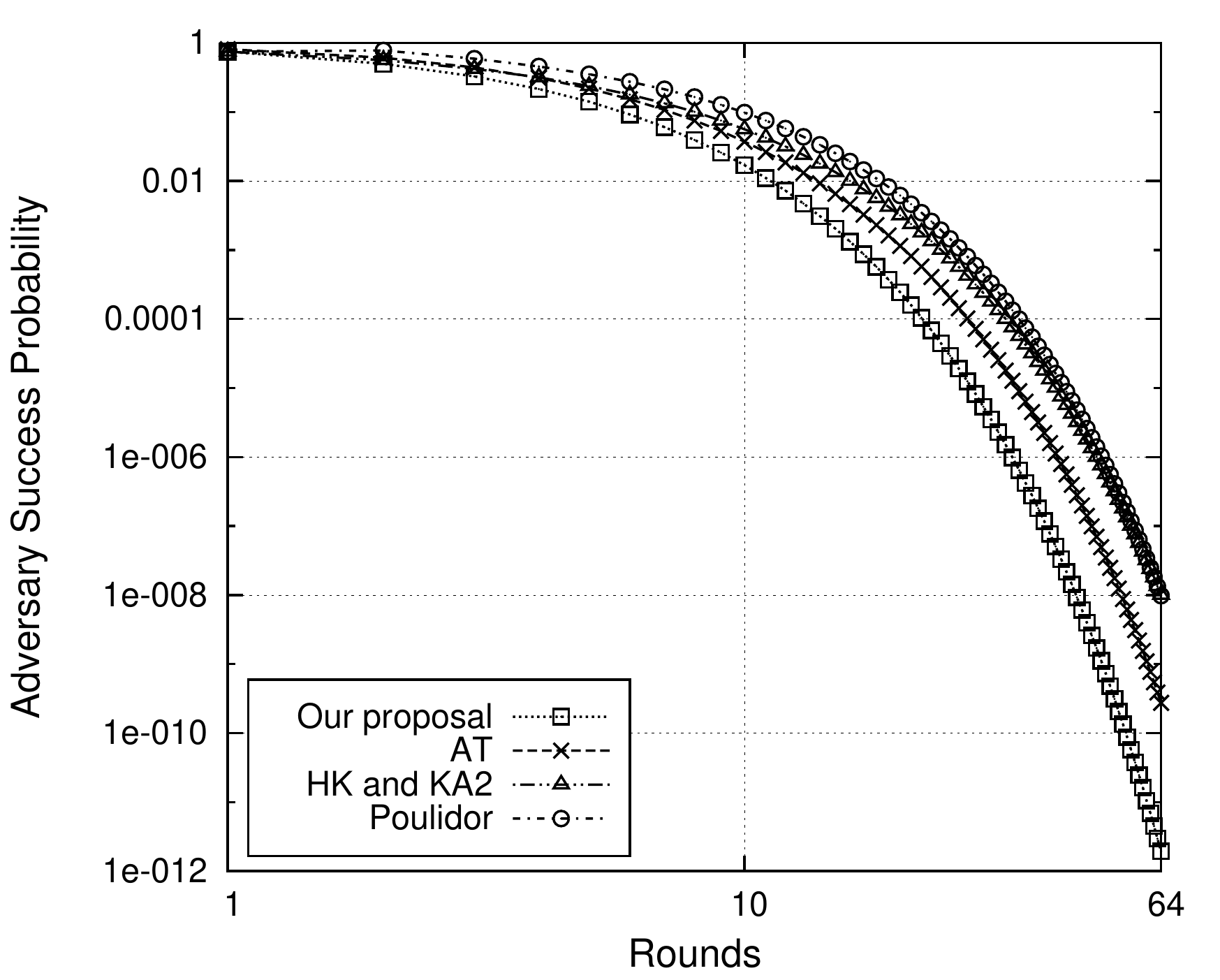}
  	\label{fig:distance}
  }
%  \vspace{-15pt}
  \caption{The mafia fraud (Figure~\ref{fig:mafia}) and distance fraud 
  (Figure~\ref{fig:distance}) success probabilities considering up to $64$ 
  rounds (logarithmic scale). The considered protocols are 
  KA2~\cite{KimA-2011-ieeetwc}, AT~\cite{AvoineT-2009-isc}, 
  Poulidor~\cite{Trujillo-Rasua:2010:PDP:1926325.1926352}, 
  HK~\cite{Hancke:2005:RDB:1128018.1128472}, and our protocol.}
%\vspace{-10pt}
\end{figure*}

Figures~\ref{fig:mafia} and~\ref{fig:distance} show that AT and KA2 are the best protocols in terms of mafia fraud while our proposal is the best in terms of distance fraud. However, it makes sense to consider the two properties together. To do so, we follow the technique used in~\cite{Trujillo-Rasua:2010:PDP:1926325.1926352} to seek for a good trade-off. This technique first discretizes the mafia fraud ($x$) and distance fraud ($y$) success probabilities. For every pair $(x,y)$, it then evaluates which protocol is the less round-consuming. This protocol is considered as the “best” for the considered pair. In case of draw between two protocols, the one that is the less memory-consuming is considered as the best protocol. Using this idea, it is possible to draw what we call a \emph{trade-off} chart, which represents for every pair $(x,y)$ the best protocol among the five considered (see Figure~\ref{fig:trade-off1}).

Figure~\ref{fig:trade-off1} shows that our protocol offers a good trade-off between resistance to mafia fraud and resistance to distance fraud, especially when a high security level against distance fraud is expected. In other words, our protocol defeats all the other considered protocols, except when the expected security levels for mafia and distance fraud are unbalanced, which is meaningless in common scenarios.

Another interesting comparison could take into consideration the memory consumption of the protocols. Indeed, for $n$ rounds of the fast phase, AT requires $2^{n+1}-1$ bits of memory, which is prohibitive for most pervasive devices. We can therefore compare protocols that require a linear memory w.r.t. the number of rounds $n$. For that, we consider a variant of AT~\cite{AvoineT-2009-isc}, denoted AT-3, that uses $n/3$ trees of depth $3$ instead of just one tree of depth $n$. By doing so, the memory consumption of all the considered protocols is below $5n$ bits of memory. The resulting trade-off chart (Figure~\ref{fig:trade-off2}) shows that constraining the memory consumption considerable reduces the area where AT is the best protocol, but it also shows that our protocol is also the best trade-off in this scenario.

\begin{figure*}
%\vspace{-10pt}
\centering
  \subfigure[Trade-off without memory constraint]{%\hspace{-5pt}
       \includegraphics[width=0.8\textwidth, angle=270]{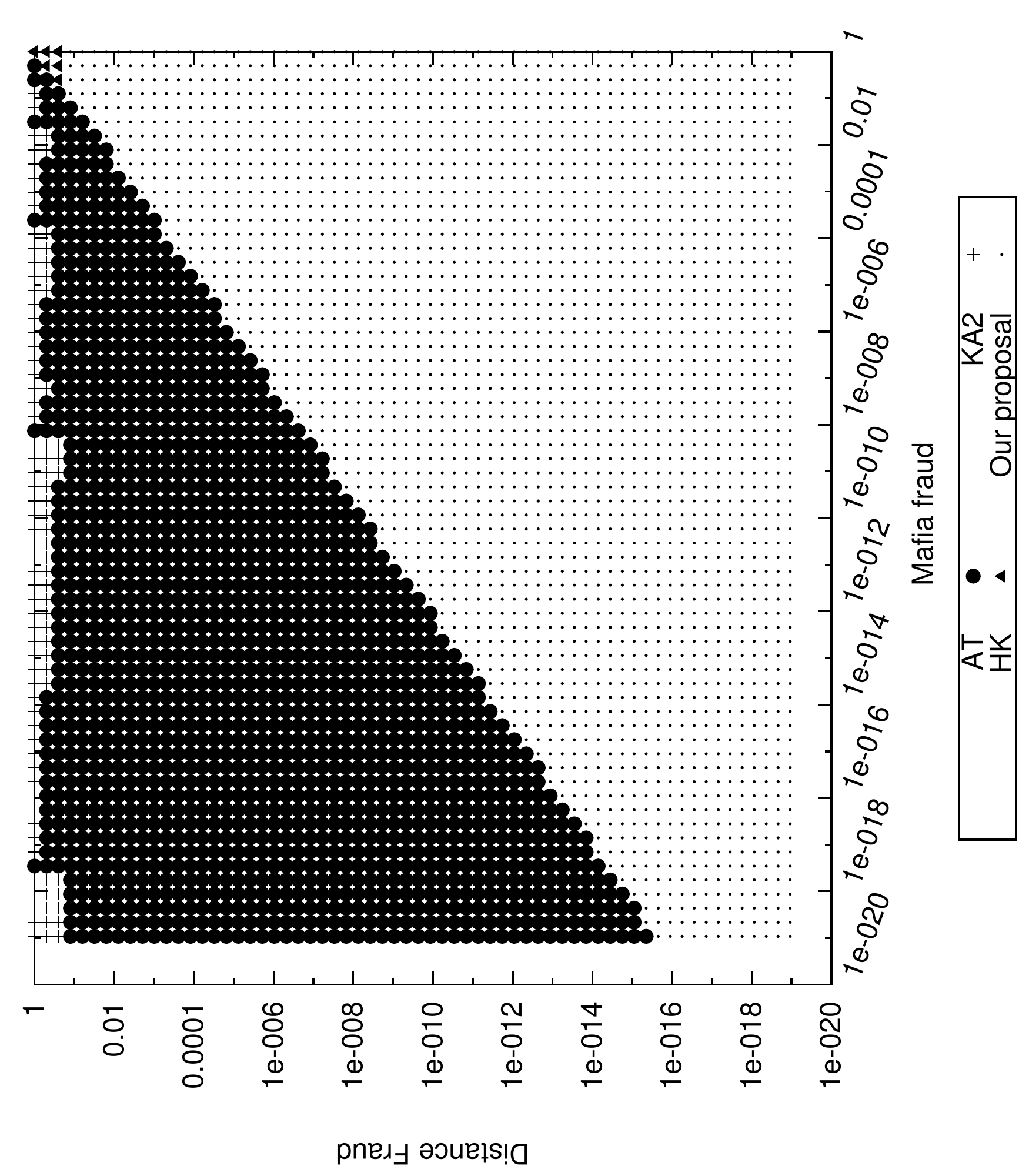}
  	\label{fig:trade-off1}
  }
  \subfigure[Trade-off with memory constraint]{%\hspace{-15pt}
       \includegraphics[width=0.8\textwidth, angle=270]{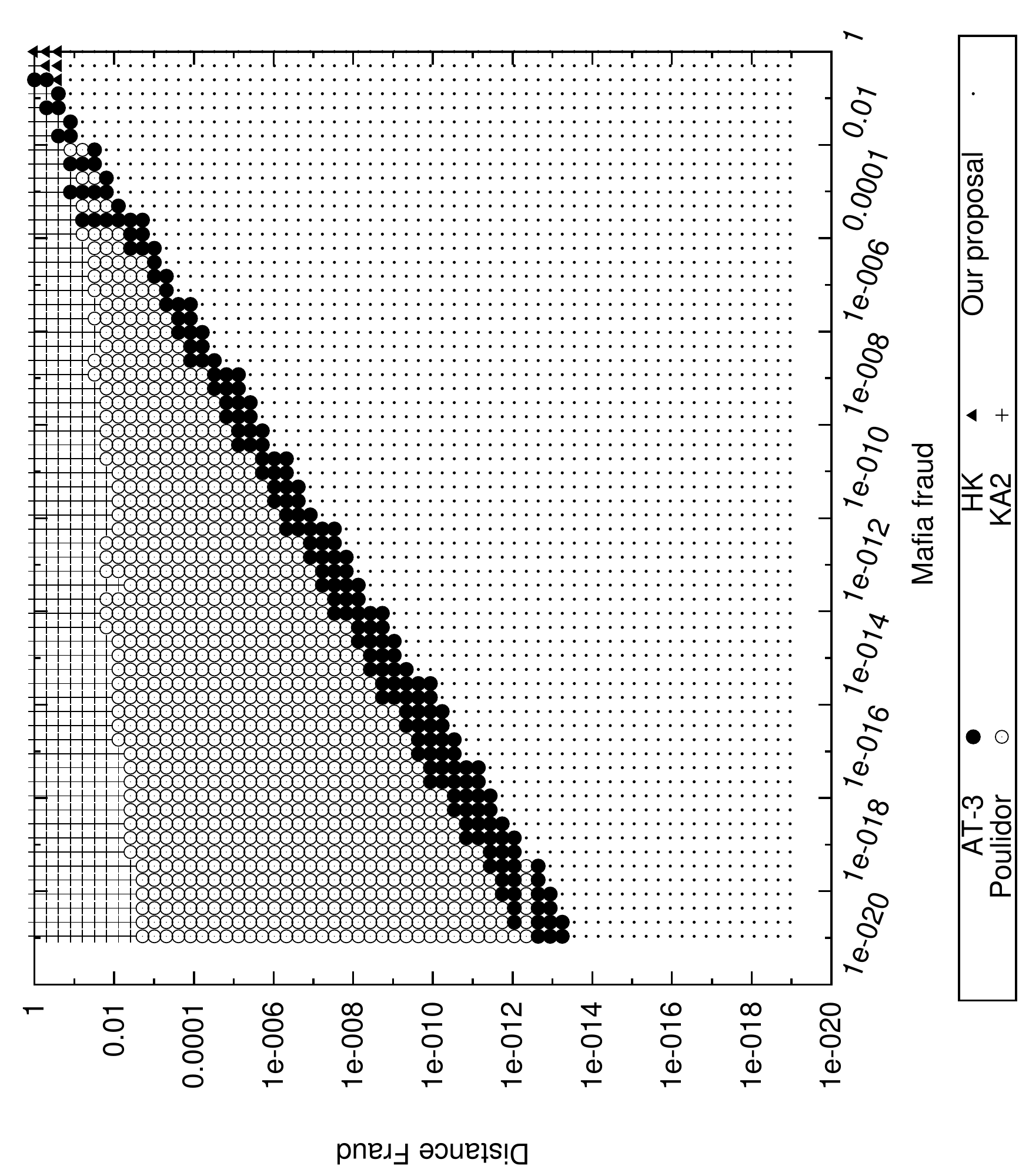}
  	\label{fig:trade-off2}
  }
%\vspace{-10pt}
  \caption{Two trade-off charts showing the most efficient protocol for each 
  pair of mafia fraud and distance fraud probability values ranging between $1$ 
  and $\left(\frac{1}{2}\right)^{64}$. Figure~\ref{fig:trade-off1} considers 
  the protocols KA2~\cite{KimA-2011-ieeetwc}, AT~\cite{AvoineT-2009-isc}, 
  Poulidor~\cite{Trujillo-Rasua:2010:PDP:1926325.1926352}, 
  HK~\cite{Hancke:2005:RDB:1128018.1128472}, and our proposal, while 
  Figure~\ref{fig:trade-off2} changes AT by its low-resource consuming variant 
  AT-3.}
%\vspace{-20pt}
\end{figure*}

\subsection{Noisy environment}

Among the protocols we are considering, only 
HK~\cite{Hancke:2005:RDB:1128018.1128472} and KA2~\cite{KimA-2011-ieeetwc} 
claim to be noise resilient. For this reason, we analyze in this section the 
performance of our proposal in the presence of noise by comparing it with HK 
and KA2. The comparison is performed by considering two properties: 
availability and security. Availability is measured in terms of false 
rejection ratio and security in terms of mafia fraud resistance. It should be 
remarked that, distance fraud resistance has not been considered for 
simplicity and because, as shown previously, our proposal outperforms all the 
other protocols in terms of this fraud. 

An important parameter when measuring availability and security for the three 
protocols is the number of allowed incorrect responses ($x$). In the case of 
KA2 and our proposal, other parameters are the minimum size of the pattern 
($s$) (denoted by $\Delta l$ in Algorithm~\ref{alg_des}) and the number of 
predefined challenges $p$ respectively. Theoretical bounds for these parameters 
in term of the number of rounds and noise probabilities might be provided, 
however, we left this non-trivial task for future work. Instead, we treat the 
three parameters $z = (x, s, p)$ as the 
variables of an optimization problem defined as follows:

\begin{problem}\label{problem}
Let $\Pi$ be a distance-bounding protocol, $n$ the number of rounds, $p_f$ the 
probability of noise in the forward channel, and $p_b$ the probability of noise 
in the backward channel. Let $\Pi_{p_f, p_b, n}^{security}(z)$ and $\Pi_{p_f, 
p_b, n}^{availability}(z)$ be the functions that, given a set of parameters $z 
= (x, s, p)$, compute the adversary success probability when mounting a mafia 
fraud against $\Pi$ and the false rejection ration of $\Pi$ respectively. Given 
a threshold $\Delta$, the optimization problem consists in:
\begin{equation*}
\begin{array}{l l}
\text{minimizing}\quad  & \quad\Pi_{p_f, p_b, n}^{security}(z) \\
\text{subject to}\quad & \quad\Pi_{p_f, p_b, n}^{availability}(z) \leq \Delta
\end{array}
\end{equation*}
\end{problem}

To follow notations of Problem~\ref{problem}, we assume that $\Pi \in 
\{\text{HK, 
KA2, Our}\}$ where ``Our'' denotes our proposal. Therefore, algorithms to 
compute 
$\Pi_{p_f, p_b, n}^{security}(z)$ and $\Pi_{p_f, p_b, n}^{availability}(z)$ 
for every $\Pi \in \{\text{HK, KA2, Our}\}$ are required. In case $\Pi = 
\text{HK}$, 
$\text{HK}_{p_f, p_b, n}^{security}(z)$ can be computed as shown 
in~\cite{Hancke:2005:RDB:1128018.1128472}, while $\text{HK}_{p_f, p_b, 
n}^{availability}(z)$ is provided by Theorem~\ref{last_theo1} (see 
Appendices). For $\Pi = \text{KA2}$, an algorithm to 
compute 
$\text{KA2}_{p_f, p_b, n}^{security}(z)$ is given in~\cite{KimA-2011-ieeetwc}. 
Unfortunately, analytical expressions for $\text{KA2}_{p_f, p_b, 
n}^{availability}(z)$, $\text{Our}_{p_f, p_b, n}^{security}(z)$, and 
$\text{Our}_{p_f, 
p_b, 
n}^{availability}(z)$, seem to be cumbersome to find. Therefore, we address 
this issue by means of simulation.

A simulation means that, given a protocol $\Pi$, all the parties (Verifier, 
Prover, and Adversary) are simulated. The protocol is then executed $10^6$ 
times and the mean of the results (either security or availability) is taken as 
the estimation.

For the experiments, we consider $48$ rounds and a false rejection ratio lower 
than $5\%$. Note that, there does not exist a real consensus on how many rounds 
should be executed in a distance bounding protocol. For instance, 
in~\cite{KimA-2011-ieeetwc} up to $40$ rounds are considered, while others 
might vary from $20$ to 
$80$. We normally choose $64$ 
rounds for our experiments~\cite{Trujillo-Rasua:2010:PDP:1926325.1926352}. 
However, since our optimization problem is solved by means of simulations whose
performance decrease with the number of rounds, we drop from $64$
to $48$ rounds.

Regarding noise we consider two cases: (a) both forward and backward channels 
introduce noise with the same probability ($p_f = p_b \leq 0.05$), and (b) the 
noise probabilities for the forward and backward channels are not necessarily 
equal and are related as follows ($p_f + p_b = 0.05$)~\footnote{We have 
performed experiments by considering several other correlations between $p_f$ 
and $p_b$. The results are not significantly different to those provided by 
these two cases, though.}. We do not consider an overall noise probability 
higher than $0.1$. 
Actually, a high noise probability makes useless all the distance bounding 
protocols proposed up-to-date.

Armed with these settings, Figure~\ref{fig:noise1} 
and Figure~\ref{fig:noise2} show the maximum resistance to mafia fraud for the 
three protocols considering the cases (a) and (b) respectively. 

%Note that, the mafia fraud resistance for both protocols has been computed by 
%optimizing the values $x$ and $s$ in the sense of Problem~\ref{problem}. 

Figure~\ref{fig:noise1} shows the mafia fraud resistance of the three protocols 
when $p_f = p_b$. As expected, the higher the noise the lower the provided 
security of the three protocols. In this scenario, however, our protocol is 
clearly the best even though KA2 achieves the highest resistance when no noise 
is considered ($p_f = p_b = 0$). 

A different scenario ($p_f + p_b = 0.05$) is shown by Figure~\ref{fig:noise2}. 
There, the security of HK improves with $p_f$, while KA2 and our protocol are 
clearly sensitive to the increase of $p_f$. This is an inherent problem of both 
protocols since a noise in the forward channel could cause a 
``desynchronization'' between the prover and the verifier. Nevertheless, thanks 
to the noise resilience mechanism proposed in Section~\ref{sec:noise}, our 
protocol deals with noise much better than KA2 and, in general, performs better 
than both HK and KA2.

\begin{figure*}
%\vspace{-10pt}
\centering
  \subfigure[$p_f = p_b$]{\hspace{-5pt}
       \includegraphics[width=0.6\textwidth, angle=270]{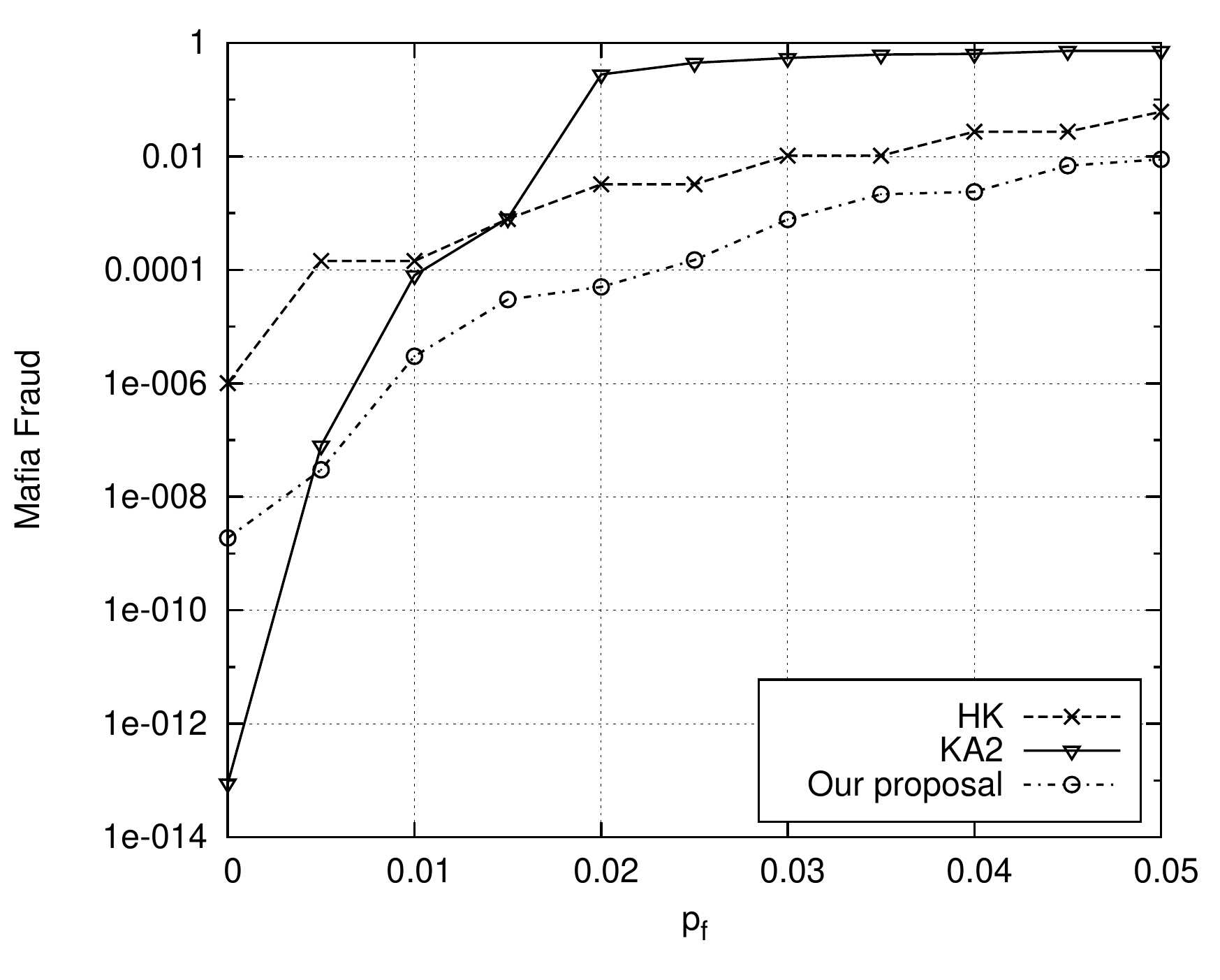}
  	\label{fig:noise1}
  }
  \subfigure[$p_f + p_b = 0.05$]{\hspace{-15pt}
       \includegraphics[width=0.6\textwidth, angle=270]{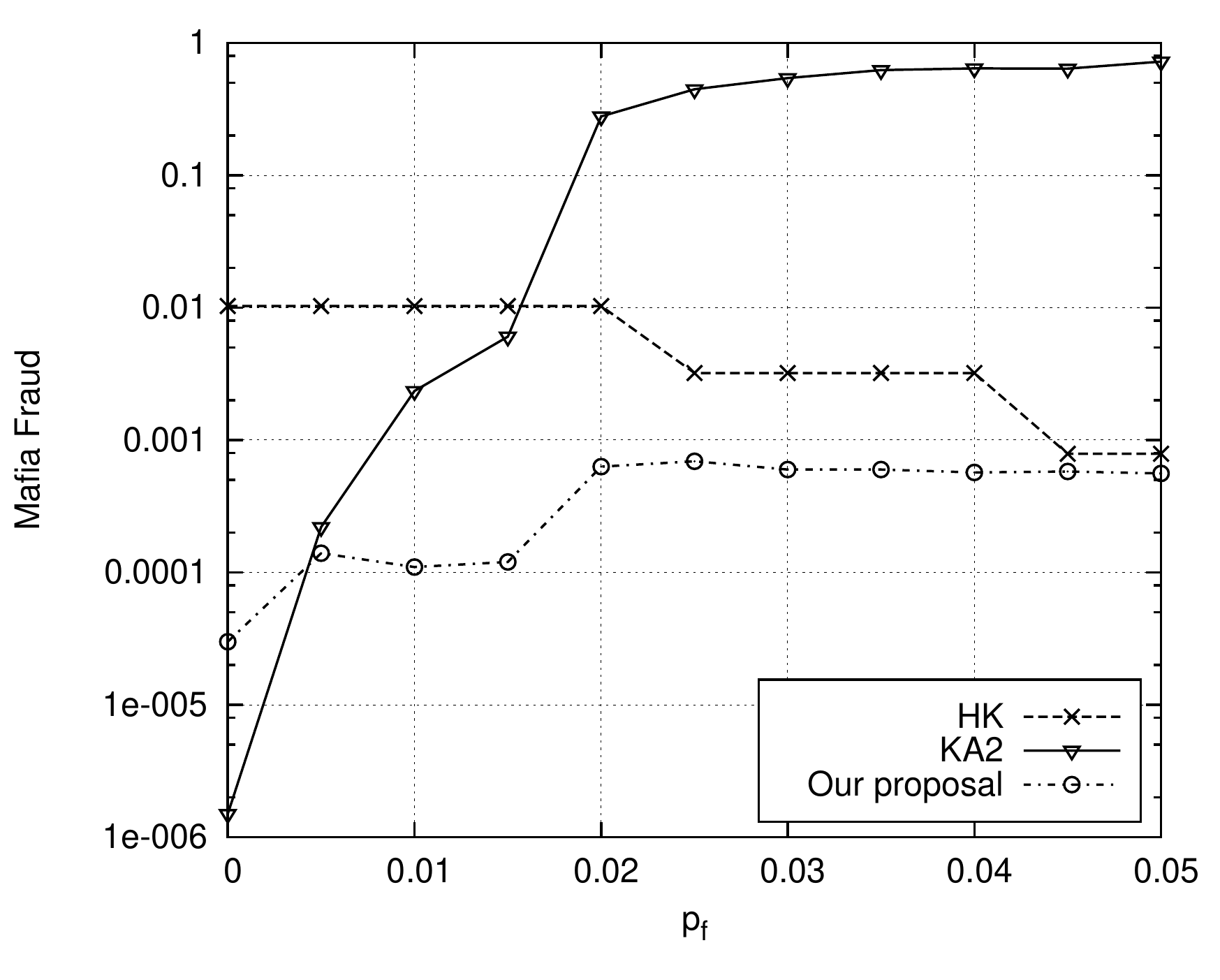}
  	\label{fig:noise2}
  }
%  \vspace{-15pt}
  \caption{The maximum resistance to mafia fraud of HK, KA2, and our proposal, 
  considering $48$ rounds, a false rejection ratio lower than $5\%$, and 
  different values $p_f$ and $p_b$: in Figure~\ref{fig:noise1} $p_f = p_b \in 
  \{0, 0.005, ..., 0.045, 0.05\}$, and in Figure~\ref{fig:noise2} $p_f + p_b = 
  0.05$ where $p_f \in \{0, 0.005, ..., 0.045, 0.05\}$.}
%\vspace{-10pt}
\end{figure*}

\section{Conclusions}\label{sec:conclusion}

A new lightweight distance-bounding protocol has been introduced in this article. The protocol simultaneously deals with both mafia and distance frauds, without sacrificing memory or requiring additional computation. The analytical expressions and experimental results show that the new protocol outperforms the previous ones. This benefit is obtained through the use of dependent rounds in the fast phase. The protocol also goes a step further by dealing with the inherent background noise on the communication channels. This is a serious advantage compared to the other existing protocols.

\section*{\appendixname}

\begin{theorem}\label{last_theo1}
Let $x$ be the maximum number of errors allowed by the verifier in the HK 
protocol. The false rejection ratio $HK_{p_f, p_b, n}^{availability}(x)$ can be 
computed as follows:
\begin{align*}
&HK_{p_f, p_b, n}^{availability}(x) = \sum_{i = 
n-x}^{n}\left(\binom{n}{i}\right.\\
&\ \ \times\left.\left(1 - \frac{p_f}{2} - p_b + p_fp_b \right)^{i}
\left(\frac{p_f}{2} + p_b - p_fp_b\right)^{n-i}\right).
\end{align*}

\end{theorem}

\begin{proof}

Let $W$ be the event that a legitimate prover's bit-answer is correct for the 
verifier. The false rejection ratio  $HK_{p_f, p_b, n}^{availability}(x)$ can 
be expressed in terms of $W$ as follows:
\begin{align}\label{theo_1_eq1}
&HK_{p_f, p_b, n}^{availability}(x) =\notag\\
& \eqspace\sum_{i = n-x}^{n}\binom{n}{i}\Pr(W)^{i}(1-\Pr(W))^{n-i}.
\end{align}

It should be noted that if a noise occurs on the forward channel then $\Pr(W) = 
\frac{1}{2}$, which happens with probability $p_f(1-p_b) + p_fp_b$. Thus:

\begin{align}\label{theo_1_eq2}
\Pr(W) & = \frac{1}{2}(p_f(1-p_b) + p_fp_b) + (1-p_f)(1-p_b) \notag\\
& = 1 - \frac{p_f}{2} - p_b + p_fp_b.
\end{align}

Equations~\ref{theo_1_eq1} and~\ref{theo_1_eq2} yield the expected result.

\end{proof}

\bibliographystyle{abbrv} %

%\bibliography{bib}

\end{document}